\documentclass{article}
\usepackage{graphicx} 
\usepackage{amsfonts,amsmath,amsthm,amssymb}
\usepackage{booktabs}
\usepackage{enumerate}
\usepackage{comment}
\usepackage{cite}

\usepackage[T1]{fontenc}\usepackage{yfonts}

\usepackage{xcolor}
\definecolor{darkred}{rgb}{0.45,0,0}
\definecolor{darkblue}{rgb}{0.2, 0.2, 0.9}
\definecolor{darkorange}{rgb}{0.8, 0.4, 0.1}

\definecolor{shadecolor}{rgb}{1,0.8,0.3}
\definecolor{myurlcolor}{rgb}{0.5,0,0}
\definecolor{mycitecolor}{rgb}{0,0,0.8}
\definecolor{myrefcolor}{rgb}{0,0,0.8}
\definecolor{hyperrefcolor}{rgb}{0.5,0,0}

\usepackage[
    colorlinks, 
    citecolor=blue, 
    urlcolor=darkred, 
    final, 
    hyperindex, 
    pagebackref, 
    linkcolor = darkblue
]{hyperref}

\newcommand{\define}[1]{{\bf \boldmath{#1}}\index{#1}}

\theoremstyle{plain}

\newtheorem{thm}{Theorem}

\newtheorem{prop}[thm]{Proposition}

\theoremstyle{remark}

\theoremstyle{definition}
\newtheorem{defn}[thm]{Definition}

\newtheorem{ex}[thm]{Example}

\newcommand{\maps}{\colon}    
\newcommand{\tr}{\operatorname{tr}}
\renewcommand{\P}{\mathbb{P}} 
\newcommand{\im}{\operatorname{im}} 

\newcommand{\R}{{\mathbb R}}  
\newcommand{\C}{{\mathbb C}}  
\newcommand{\Z}{{\mathbb Z}}  
\renewcommand{\H}{{\mathbb H}}  
\renewcommand{\O}{{\mathbb O}}  

\newcommand{\U}{{\rm U}}    
\newcommand{\SO}{{\rm SO}}    
\newcommand{\Sp}{{\rm Sp}}    
\newcommand{\SU}{{\rm SU}}    
\newcommand{\SL}{{\rm{SL}}}  
\newcommand{\Spin}{{\rm Spin}}    
\newcommand{\E}{{\rm E}}       
\newcommand{\F}{{\rm F}}       
\newcommand{\G}{{\rm G}}       
\renewcommand{\S}{\mathrm{S}} 
\newcommand{\Inn}{\mathrm{Inn}}  

\newcommand{\so}{{\mathfrak{so}}}  
\newcommand{\su}{{\mathfrak{su}}}  
\newcommand{\gl}{{\mathfrak{gl}}}  
\renewcommand{\u}{{\mathfrak{u}}}  
\newcommand{\e}{{\mathfrak{e}}}   
\newcommand{\g}{{\mathfrak{g}}}  
\newcommand{\p}{{\mathfrak{p}}}  
\renewcommand{\k}{{\mathbf{k}}}  
\newcommand{\inn}{{\mathfrak{inn}}} 

\newcommand{\M}{\mathrm{M}}   
\newcommand{\h}{\mathfrak{h}}  

\newcommand{\SM}{\mathrm{SM}}

\makeatletter
\newcommand*\bigcdot{\mathpalette\bigcdot@{.5}}
\newcommand*\bigcdot@[2]{\mathbin{\vcenter{\hbox{\scalebox{#2}{$\m@th#1\bullet$}}}}}
\makeatother

\date{\today}

\title{Jordan Pair Quantum Theory \\ and the Standard  Model}

\author{John C. Baez\footnote{School of Physics and Astronomy and School of Mathematics, University of Edinburgh and Department of Mathematics, U.\ C.\ Riverside} \and
Endre Bokor\footnote{School of Mathematics and School of Physics and Astronomy, University of Edinburgh} \and
Latham Boyle\footnote{Higgs Centre for Theoretical Physics, School of Physics and Astronomy, University of Edinburgh, and Perimeter Institute for Theoretical Physics} 
}

\date{\today}

\begin{document}

\maketitle 

\begin{abstract}
    Jordan pairs and hermitian Jordan triples were discovered by mathematicians studying Jordan algebras, which describe the possible algebras of observables in quantum mechanics.  We point out a striking correspondence between the 
    doubly exceptional hermitian Jordan triple (the so-called ``bi-Cayley'' triple) 
    and the structure of the Standard Model of particle physics.  We also point out how ordinary quantum mechanics may be reformulated, and generalized, using hermitian Jordan triples.
\end{abstract}

\tableofcontents

\section{Introduction}
The goal of this paper is to explain how quantum mechanics can be done using an algebraic structure called a hermitian Jordan triple, and show how some aspects of the Standard Model arise naturally from a hermitian Jordan triple built using the complexified octonions $\C \otimes \O$, also known as the `bioctonions'.  

When Jordan, von Neumann and Wigner \cite{Jordan:1933vh} introduced Euclidean Jordan algebras to formalize of the concept of `observables' in quantum physics, they classified these algebras.  They come in several infinite series, but there is also one exception: $\h_3(\O)$, consisting of $3 \times 3$ self-adjoint octonion matrices.  Ever since then, physicists have wondered whether this algebra plays a role in fundamental physics \cite{Gunaydin1978,Townsend1985,Goddard1987,gursey1996role,Manogue:2009gf,Dubois-Violette:2016kzx,Liebmann2019}.  The work of Todorov and Dubois-Violette \cite{Todorov:2018mwd} took a promising step by showing that the Standard Model gauge group is a subgroup of $\F_4$, the automorphism group of $\h_3(\O)$.    Later work showed that the Standard Model gauge group could be neatly described using Jordan subalgebras of $\h_3(\O)$.  Namely, it is precisely the largest connected subgroup of $\F_4$ that preserves a copy of $\h_3(\C)$ in $\h_3(\O)$, and a copy of $\h_2(\C)$ inside that \cite{Baez:2026tbe}.   This suggests an interesting interplay between octonionic quantum mechanics and ordinary complex quantum mechanics.

However, this line of work has not managed to explain the chiral structure of the fermions in the Standard Model.  To do that, one of the present authors \cite{Boyle:2020ctr} argued one should replace the octonions by the complexified octonions $\O_\C = \C \otimes \O$, also called the `bioctonions'.  A single generation of fermions, including a right-handed neutrino, can be identified with $\O_\C^2$.  Moreover, $\O_\C^2$ is the tangent space of a compact symmetric space on which the group $\E_6$ acts transitively as isometries.   This space is called the `bioctonionic Rosenfeld plane', or $\O_\C\P^2$.   The stabilizer of any point of the bioctonionic Rosenfeld plane is $\Spin(10) \cdot \U(1)$, where $\U(1)$ acts as multiplication by phases, and $\Spin(10)$ acts on the tangent space $(\C \otimes \O)^2$ just as it acts on one generation of fermions in an $\SO(10)$ GUT \cite{Georgi:1974my, Fritzsch:1974nn}.  As pointed out by Krasnov \cite{Krasnov:2022meo,Krasnov2}, the Standard Model gauge group is then the subgroup of $\Spin(10)$ giving transformations of $\O_\C^2$ that preserve some natural structure on this space.

Unfortunately, there is no Euclidean Jordan algebra underlying the space $\O_\C\P^2$ in the way $\h_3(\O)$ underlies $\O\P^2$.  Thus, this approach to the Standard Model might seem to have lost the clear connection to the foundations of quantum physics that we see in the exceptional Jordan algebra.  However, we show here that a slight generalization of Jordan algebras explains what is going on.

Starting in the late 1940s, researchers on Jordan algebras isolated two important generalization of Jordan algebras, called Jordan triples \cite{Jacobson1949} and Jordan pairs \cite{Loos1975,Loos1977}.  As we shall explain in detail, positive hermitian Jordan triples are equivalent to positive hermitian Jordan pairs, and both correspond to compact hermitian symmetric spaces.  Such spaces had been classified by Cartan much earlier \cite{Cartan1935}: there are four infinite series and two exceptions, one related to the Jordan triple $\h_3(\O_\C)$ and one related to the Jordan triple $\O_\C^2$.  The first of these gives rise to a hermitian symmetric space of complex dimension 27, while the second gives rise to $\O_\C\P^2$.

Since $\O_\C\P^2$ is connected not to a Euclidean Jordan algebra but to a positive hermitian Jordan triple, our first goal here is to explain how to do quantum physics with positive hermitian Jordan triples (or equivalently, pairs).   Our second goal is to apply the resulting theory to the Standard Model, linking this theory to the two exceptional Jordan triples, and especially $(\O_\C)^2$.

Our work has antecedents.  For example, in 1978 Biedenharn and Truini attempted to define the basic concepts of quantum mechanics---states, observables, and symmetries---starting from any Jordan triple system \cite{BiedenharnTruini1982a}.  They also applied these ideas to particle physics using $\h_3(\O_\C)$ \cite{BiedenharnTruini1982b, TruiniOlivieriBiedenharn1986}.   However, their formulation of quantum mechanics was different than ours, and they did not make the connection to the Standard Model that we explain here.  In 2020 Nasmith studied a series of nested 3-graded Lie algebras and used them to obtain the Standard Model gauge group and some aspects of its representation on known particles \cite{Nasmith2020,Nasmith2023}.  These nested 3-graded Lie algebras have remarkable properties: for example they arise from taking the Dynkin diagram for $\E_7$ and removing dots one at a time in a specific pattern, studied earlier by Manin in his work on del Pezzo surfaces \cite[Sec.\ 25.5.7]{Manin1986}, and later related to hermitian symmetric spaces \cite{SerganovaSkorobogatov2007}.  However, the physical significance of all this mathematics remained unclear.  In fact these 3-graded Lie algebras are tightly connected to positive hermitian Jordan triples, which can serve as the basis for a generalization of quantum theory---and the particular nesting pattern, where one Jordan triple fits inside another, plays a key role in our work.

\vskip 1em

\textbf{Plan of the Paper.}  In Section \ref{sec:jordan_pairs} we start with a quick introduction to hermitian Jordan pairs and triples, their relation to 3-graded Lie algebras and hermitian symmetric spaces, and their classification.  In Section \ref{sec:quantum} we describe a generalization of ordinary quantum theory based on positive hermitian Jordan pairs and triples.  In Section \ref{sec:exceptional} we discuss the two exceptional hermitian Jordan triples: $h_3(\O_\C)$, which we call the Albert triple, and $\O_\C^2$, which we call the bi-Cayley triple.   In Section \ref{sec:standard} we apply all these ideas to the Standard Model.

\section{Hermitian Jordan pairs and Jordan triples}
\label{sec:jordan_pairs}

During the 1960s and 70s, mathematicians realized that it was natural to pivot from the notion of Jordan algebra to a variety of more flexible concepts, including hermitian Jordan pairs and hermitian Jordan triples \cite{Loos1975}.  In this section, we introduce these structures and their connection to 3-graded Lie algebras and hermitian symmetric spaces, which play an important role in sections to come.

\subsection{Basic definitions}

We begin by defining hermitian Jordan pairs and hermitian Jordan triples.  In what follows all vector spaces will be assumed finite-dimensional, and complex unless otherwise specified.  We write $V$ for a pair of vector spaces $(V_+, V_-)$ over the complex numbers.   We use $\sigma$ to mean $+$ or $-$, and then $-\sigma$ means $-$ or $+$.  Given pairs of vector spaces $V$ and $W$, we write $f \maps V \to W$ for a pair of linear maps $f_\sigma \maps V_\sigma \to W_\sigma$.  We also write $g \maps V \times V \times V \to V$ for a pair of trilinear maps $g_\sigma \maps V_\sigma \times V_{-\sigma} \times V_\sigma \to V_\sigma$: the key example of this is the triple product below.

\begin{defn} A \define{Jordan pair} is a pair of vector spaces $V=(V_+,V_-)$ equipped with a \define{triple product}
\begin{equation}
    \{\;\bigcdot\;,\;\bigcdot\;,\;\bigcdot\;\} \maps V_{\sigma} \times V_{-\sigma} \times V_{\sigma} \to V_{\sigma}.
\end{equation}
Given an element $v\in V$, by which we mean a pair $v=(v_+,v_-)$ with $v_\sigma \in V_\sigma$, we define $\delta_{v} \maps V\to V$ to be the pair of linear maps $\delta_{v}^{\sigma} \maps V_{\sigma}\to V_{\sigma}$ given by 
\begin{equation}
  \label{eq:delta}
    \delta_{v}^{\sigma}=\sigma \{v_{\sigma},v_{-\sigma},\;\bigcdot\;\}.
\end{equation}
Finally, we take these ingredients to satisfy the following two axioms:
\begin{subequations}
  \begin{eqnarray}
    \{x,y,z\} &=&\{z,y,x\} \label{eq:symmetry} \\
    \delta_{v}\{x,y,z\} &=&\{\delta_{v}x,y,z\} +\{x,\delta_{v}y,z\} +\{x,y,\delta_{v}z\}.\label{eq:derivation}
  \end{eqnarray}  
\end{subequations}
\end{defn}

Some comments on this definition are in order.  First, note that the outer two arguments of the triple product are drawn from one space, $V_{\sigma}$, while the middle argument is drawn from the other space $V_{-\sigma}$, and the result lies the democratically-preferred space $V_{\sigma}$.  Equation \eqref{eq:delta} can be written more explicitly as
\begin{equation} 
\delta_{v}w_{+}=\{v_{+},v_{-},w_{+}\} , \qquad
\delta_{v}w_{-}=-\{v_{-},v_{+},w_{-}\} .
\end{equation}
Axiom \eqref{eq:symmetry} says that the triple product is symmetric under reversing its arguments, while axiom \eqref{eq:derivation} says that $\delta_v$ acts as a derivation.  
The second axiom is usually presented in a number of other ways \cite{McCrimmon2004, Loos1975}, which are equivalent to ours.  For example, if we avoid the use of $\delta$, it becomes
\[  \{v_+,v_-,\{x,y,z\}\} = \{\{v_+,v_-,x\},y,z\}- 
\{x,\{v_-,v_+,y\},z\} + \{x,y,\{v_+,v_-,z\}\} 
\]
where $x,z \in V_\delta$ and $y \in V_{-\delta}$.

\begin{ex}  
\label{ex:Mpq}
For the simplest example of a Jordan pair, let $V_+ = \M_{p,q}(\C)$ be the space of $p \times q$ complex matrices, let $V_- = \M_{q,p}(\C)$, and let
\[    \{x,y,z\} = \frac{1}{2}(x y z + z y x) \]
for $x,z \in V_\sigma$, $y \in V_{-\sigma}$.  Axiom \eqref{eq:symmetry} is clear and axiom \eqref{eq:derivation} can be checked with a calculation. 
\end{ex}

\begin{defn} A \define{hermitian Jordan pair} is a Jordan pair $V = (V_+, V_-)$ equipped with an antilinear map $\kappa \maps (V_+, V_-) \to (V_-, V_+)$ that is an involution in the following sense:
\begin{equation}
\label{eq:kappa_involution}
    \kappa_\sigma \kappa_{-\sigma} = 1
\end{equation}
and preserves the triple product in the following sense:
\begin{equation}
\label{eq:kappa_preserves_triple_product}
\kappa_\sigma \{x,y,z\} = \{\kappa_\sigma x, \kappa_{-\sigma} y, \kappa_{\sigma} z \}  
\end{equation}
for all $x,z \in V_\sigma$, $y \in V_{-\sigma}$. 
\end{defn}

It will cause no ambiguity to abbreviate 
equations \eqref{eq:kappa_involution} and \eqref{eq:kappa_preserves_triple_product} as follows:
\begin{equation}
\kappa^2 = 1, \qquad
\kappa \{x,y,z \} = \{\kappa x, \kappa y, \kappa z\}.
\end{equation}

\begin{ex}
The Jordan pair of Example \ref{ex:Mpq} becomes hermitian if we define $\kappa \maps V_\sigma \to V_{-\sigma}$ by $\kappa u = u^\dagger$, where $\dagger$ stands for the conjugate transpose matrix.
\end{ex}

Given a hermitian Jordan pair, we can use $\kappa$ to identify $V_+$ and $V_-$ and express all the structure in terms of $V_+$.   This leads to the definition of `hermitian Jordan triple'.  The first step is to define a new map
\begin{equation}
[\;\bigcdot\;,\;\bigcdot\;,\;\bigcdot\;]\maps V_+ \times V_+ \times V_+ \to V_+ 
\end{equation}   
given by
\begin{equation}
    [x,y,z]=\{x,\kappa_+ y,z\}.
\end{equation}
This new triple product is \define{skew-linear}, by which we mean linear in its first and last arguments and antilinear in its middle argument.  The axioms for a hermitian Jordan pair can then be expressed in terms of this skew-linear triple product on $V_+$, and these give the definition of a hermitian Jordan triple system---or `hermitian Jordan triple' for short.   

\begin{defn}
A \define{hermitian Jordan triple} is a complex vector space $W$ equipped with a skew-linear map
\[  [\;\bigcdot\;,\;\bigcdot\;,\;\bigcdot\;] \maps W \times W \times W \to W \]
obeying 
\begin{subequations}
  \begin{eqnarray}
    [x,y,z] &=& [z,y,x] \label{eq:hermitian_symmetry} \\
    {[a,b,[x,y,z]]} &=& [[a,b,x],y,z] - [x,[b,a,y],z] + [x,y,[a,b,z]]    \label{eq:hermitian_derivation}
  \end{eqnarray}  
\end{subequations}
\end{defn}

By how the definitions are set up, hermitian Jordan triple systems are equivalent to hermitian Jordan pairs: we can pass freely from one to the other, and use whichever seems most convenient at the time.

\begin{ex} 
 $\M_{p,q}(\C)$ is a hermitian Jordan triple
with skew-linear triple product 
\[  [x,y,z] = \frac{1}{2}(x y^\dagger z + z y^\dagger x).\]
\end{ex}

This example obeys a positivity condition that turns out to be important more generally, because it allows us to introduce inner products into Jordan pair quantum theory.  Given a hermitian Jordan triple $W$, any pair of elements $x,y \in W$ defines a linear map $[x,y, \; \bigcdot \;] \maps W \to W$, and taking the trace of this we define
\begin{equation}
\label{eq:inner_product}
\langle x \vert y \rangle = \text{Tr}[x,y, \; \bigcdot \; ] . 
\end{equation}
This is a sesquilinear form on $W$: linear in the first argument, antilinear in the second.

\begin{defn}
\label{defn:positive}
A hermitian Jordan triple $W$ is \define{positive} if its sesquilinear form is an inner product, i.e.
\[    x \ne 0 \implies \langle x \vert x \rangle > 0.\]
We say a hermitian Jordan pair is positive if its corresponding hermitian Jordan triple is positive.
\end{defn}

In Section \ref{subsec:symmetric_spaces} we explain the classification of positive hermitian Jordan triples, and how they give hermitian symmetric spaces.

\subsection{Automorphisms and derivations}
\label{subsec:derivations}

An \define{automorphism} $U \maps V \to V$ of a Jordan pair $V$ is a pair of invertible linear maps $U_{\sigma}\maps V_{\sigma}\to V_{\sigma}$ satisfying
\begin{equation}
  \label{Jordan_pair_automorphism}
  U\{x,y,z\}=\{Ux,Uy,Uz\}.
\end{equation}
An infinitesimal generator of automorphisms is called a \define{derivation}.  Writing $U=\exp(t\delta)$ for some $\delta \maps V \to V$, differentiating equation \eqref{Jordan_pair_automorphism} with respect to $t$, and setting $t = 0$, we see that a derivation $\delta$ is a map $\delta \maps V \to V$ obeying
\begin{equation}
  \delta\{x,y,z\}=
  \{\delta x,y,z\}
  +\{x,\delta y,z\}
  +\{x,y,\delta z\}.
\end{equation}
Conversely, any $\delta \maps V \to V$ obeying this equation is a derivation.   The maps $\delta_{v} \maps V \to V$ defined in \eqref{eq:delta} are examples of derivations.  In fact the vector space spanned by such derivations is a Lie algebra, because one can check that
\begin{equation}
\label{eq:bracket_of_inner_derivations}
[\delta_v, \delta_w] = \delta_{(\delta_{v}w_{+},w_{-})}+\delta_{(w_{+},\delta_{v}w_{-})}. 
\end{equation}
We denote this Lie algebra as $\inn(V)$, and call its elements \define{inner derivations} of $V$.   We call the Lie subgroup of automorphisms of $V$ generated by inner derivations the group of \define{inner automorphisms} of $V$, $\Inn(V)$.

A \define{real automorphism} of a hermitian Jordan pair $V$ is defined to be a map $U \maps V \to V$ satisfying (\ref{Jordan_pair_automorphism}) and commmuting with the antilinear involution $\kappa \maps (V_+,V_-) \to (V_-,V_+)$ that maps $(v_+,v_-)$ to $(\kappa v_-, \kappa v_+)$.  For short, we say $U$ obeys
\begin{equation}
 U \kappa = \kappa U .
\end{equation}

Since $\kappa$ is antilinear, the Lie group of real automorphisms of $V$ is a real Lie group, and its Lie algebra is a real Lie algebra.  The same is true for the Lie algebra of the group of real inner automorphisms,
\[    \inn_\R(V) = \{ \delta \in \inn(V) \; \vert \; \kappa \delta = \delta \kappa \} .\]
We call this the Lie algebra of \define{real inner derivations} of $V$.

We can get some real inner derivations of a hermitian Jordan pair as follows:

\begin{prop}
\label{prop:Inn_R}
Suppose $V$ is a hermitian Jordan pair with antilinear involution $\kappa$.   For any $v = (v_+, v_-) \in V_+ \times V_-$, the map
\begin{equation}
\label{eq:Inn_R}
  D_v = \delta_{v}-\delta_{\kappa v}
\end{equation}
is an element of $\inn_\R(V)$.
\end{prop}

\begin{proof}
If $\delta$ is a derivation of $V$ so is $\kappa \delta \kappa$, since
\[    (\kappa \delta \kappa)\{x,y,z\} = 
 \kappa \delta \{\kappa x, \kappa y, \kappa z\} = \]
\[   \{\kappa \delta \kappa x,  y,  z\} + 
    \{ x,  \kappa \delta \kappa y,  z\} + 
    \{ x,   y,  \kappa \delta \kappa z\} \]
where we use $\kappa^2 = 1$.   Thus $\delta + \kappa \delta \kappa$ is a derivation, and it commutes with $\kappa$ because
\[  \kappa (\delta + \kappa \delta \kappa) = \kappa \delta + \delta \kappa
= (\delta + \kappa \delta \kappa) \kappa .\]  
We know $\delta_v$ is an inner derivation, and $\kappa \delta_v \kappa$ is also inner since for $(x_+,x_-) \in V_+ \times V_-$ we have
\[ 
\begin{array}{ccl}
(\kappa \delta_v \kappa)(x_+,x_-) &=& \kappa \delta_v (\kappa x_-, \kappa x_+) \\ [3pt]
&=& \kappa \left(\{v_+, v_-, \kappa x_-\} , - \{v_-, v_+, \kappa x_+\}\right) \\ [3pt]
&=& \left(-\{\kappa v_-, \kappa v_+, x_+\}, 
      \{\kappa v_+, \kappa v_-, x_-\}\right)\\ [3pt]
&=& - \delta_{\kappa v} (x_+,x_-)
\end{array} 
\]
 It follows that $\delta_v + \kappa \delta_v \kappa = \delta_v - \delta_{\kappa v}$ is an inner derivation that commutes with $\kappa$.  
\end{proof}

We conjecture, on the basis of very little evidence, that $\inn_\R(V)$ is actually spanned by derivations of the form $D_v$.

Since hermitian Jordan pairs correspond to hermitian Jordan triples, the whole theory of automorphisms and derivations carries over to hermitian Jordan triples.  Suppose $W$ is a hermitian Jordan triple.   An \define{automorphism} of $W$ is an invertible linear map $U \maps W \to W$ such that 
\begin{equation}  
U[x,y,z] = [Ux,Uy,Uz] .
\end{equation}
A \define{derivation} of $W$ is a linear map $\delta \maps W \to W$ such that
\begin{equation}
\delta[x,y,z] = [\delta x, y, z] + [x, \delta y, z] + [x, y, \delta z] .
\end{equation}
Any inner derivation $\delta_v$ on a hermitian Jordan pair $V = (V_+, V_-)$ corresponds to a derivation on the corresponding hermitian Jordan triple $W = V_+$.  We call the linear span of such derivations the space of \define{inner derivations} of $W$, $\inn(W)$.  This forms a Lie algebra by equation \eqref{eq:bracket_of_inner_derivations}.
We call the Lie group of transformations of $W$ generated by this Lie algebra the group of \define{inner automorphisms} of $W$, and denote it as $\Inn(W)$.  Similarly, we can define real inner derivations of $W$, and say the Lie subgroup of real automorphisms of $W$ generated by these is the group of \define{real} inner automorphisms, $\Inn_\R(W)$.  We call its Lie algebra $\inn_\R(W)$.

We can transfer the real inner derivations $D_{v,v'}$ of a hermitian Jordan pair $V$ given by Proposition \ref{prop:Inn_R} to inner derivations of the corresponding Jordan triple $W = V_+$, and we get the following:

\begin{prop}
Suppose $W$ is a hermitian Jordan triple.  For any pair $w, w' \in W$ the map
\begin{equation}
D_{w,w'} = [w,w', \; \bigcdot \; ] - [w', w, \; \bigcdot \; ].
\end{equation}
is an element of $\inn(V)$ that is skew-adjoint in the following sense:
\begin{equation}
\langle D_{w,w'} x \; \vert \; x' \rangle =
-\langle x \; \vert \; D_{w,w'} x' \rangle 
\end{equation}
\end{prop}

\subsection{Tripotents and the Peirce decomposition}
\label{subsec:tripotents}

In the theory of associative algebras and Jordan algebras, elements with $e^2 = e$ are very important: these are called idempotents.  In the theory of Jordan triples, `tripotents' play a similar role.  Here we briefly sketch the theory of these; for details and full proofs see Loos \cite{Loos1977}, Chu \cite{Chu2012} and Roos \cite{Roos2000}.

Suppose $W$ is a positive hermitian Jordan triple.
A \define{tripotent} is an element $e \in W$ with 
\begin{equation}
[ e, e , e ] = e .
\end{equation}
Unlike an idempotent, we can multiply a tripotent by a phase (a complex number with norm 1) and obtain a new tripotent.  This fact is important in applications to quantum theory: indeed, in Section \ref{sec:quantum} we shall see a positive hermitian Jordan triple whose ``minimal'' tripotents are exactly the unit vectors in $\C^n$.  Thus, minimal tripotents, to be defined below, are a generalization of pure states.

If $e$ is a tripotent, the eigenvalues of the operator
$[e,e,\, \bigcdot \, ]$ lie in the set $\{0, \tfrac{1}{2}, 1\}$.  We call the eigenspaces of this operator the \define{Peirce spaces} of $e$:
\begin{equation}  
E_\lambda(e) = \{w \in W \; \vert \; \{e, e, w\} = \lambda w \}.
\end{equation}
More specifically we call $E_0(e)$, $E_{1/2}(e)$ and $E_1(e)$ the Peirce $0$-space, Peirce $1/2$-space and Peirce $1$-space of $e$, respectively.   Peirce spaces are a powerful tool for understanding the structure of a positive hermitian Jordan triple, in part because of this rule:
\begin{equation}
\label{eq:peirce_multiplication_rule}
x \in E_i(e), \, y \in E_j(e), \, z \in E_k(e) \implies [x,y,z] \in E_{i-j+k}(e) .
\end{equation}
This implies $[x,y,z] = 0$ if $i-j+k \notin \{0,1/2,1\}$.  It also implies that each Peirce space is closed under the triple product, and thus a hermitian Jordan triple system in its own right.

Suppose $e, f \in W$ are two tripotents.  One can show that $e$ is in the Peirce $0$-space of $f$ if and only if $f$ is in the Peirce $0$-space of $e$.   In this case we say $e$ and $f$ are \define{orthogonal}.   A fundamental result on positive hermitian Jordan triple systems is that every $w \in W$ can be written as
\[   w = \lambda_1 e_1 + \cdots \lambda_n e_n \]
for some $\lambda_i > 0$ and a collection of mutually orthogonal tripotents $e_i$.   This is sometimes called the \define{spectral theorem}.

There is a partial ordering on tripotents: given tripotents $e, f \in W$ we say $e \le f$ if there is a tripotent $e'$, orthogonal to $e$, such that $f = e + e'$.  A tripotent $f$ is called \define{minimal} if the only tripotent $e$ with $e \le f$ is $e = 0$.  A \define{frame} for $W$ is a maximal set of mutually orthogonal minimal tripotents.  One can think of a frame as similar to an orthonormal basis.  In particular, every frame for $W$ has the same cardinality, which is finite.  This cardinality is called the \define{rank} of $W$.

We define the \define{rank} of a tripotent $e$ to be the length $n$ of the longest chain of tripotents $0 < e_1 < \cdots < e_n = e$.  For example, a minimal tripotent has rank $1$, while the largest possible rank of a tripotent in $W$ is the rank of $W$ itself.  It is easy to see that the rank of a tripotent does not change when we multiply it by a phase.  Less obviously, a minimal tripotent is the same as a tripotent whose Peirce 1-space has dimension 1.

The concepts introduced here interact in some beautiful ways. In a \define{simple} hermitian Jordan triple---i.e., one that is not a direct sum of two others in a nontrivial way---any rank $r$ tripotent $e$ can be mapped to any other using the group of inner automorphisms \cite[Corollary 5.12]{Loos1977}.  Moreover, if $W$ is a simple hermitian Jordan triple, then so is each of its Peirce spaces \cite{Mccrimmon1979Peirce}.


We also need some explicit formulas for operators that project $W$ onto the Peirce subspaces of a tripotent \cite{edwards1996range}.  If we define the operators
\begin{subequations}
  \begin{eqnarray}
    Q_{x}y&=&[x,y,x] \\
    D(x,y)z&=&[x,y,z]
  \end{eqnarray}
\end{subequations}
on a hermitian Jordan triple, and pick an arbitrary tripotent $e$, the projection operators onto the $0, 1/2$ and $1$ Peirce spaces of $e$ are given explicitly as follows:
\begin{subequations}
  \label{Peirce_projectors}
  \begin{eqnarray}
   P_0(e)&=& 1_{W} -2D(e,e)+Q_e^2 \\
    P_{1/2}(e)&=&2\left(D(e,e)-Q_e^2\right) \\ 
        P_1(e)&=&Q_e^2 
\end{eqnarray}
\end{subequations}
  
\subsection{Reinterpretation as graded Lie algebras}
\label{subsec:graded}

Jordan pairs and Jordan triples may seem exotic at first, but they are closely related to some less exotic structures: graded Lie algebras.  This is crucial for understanding their applications to geometry.  Building on earlier work by Tits--Kantor--Koecher \cite{Tits1962,Kantor1964,Koecher1967} and Meyberg \cite{Meyberg1970}, Loos \cite{Loos1975} showed how to construct a Jordan pair $V$ from a \define{3-graded Lie algebra}, meaning a Lie algebra $\g$ with a $\Z$-grading such that.
\begin{equation}
    \g=\g_{-1}\oplus\g_{0}\oplus\g_{+1}.
\end{equation}
Saying that $\g$ is $\Z$-graded implies that $[\g_{i},\g_{j}]\in\g_{i+j}$.  The construction is simply to take
\begin{equation}  
\label{eq:Jordan_from_3-graded_1}
V_+ = \g_1, \quad V_- = \g_{-1} \end{equation}
and define
\begin{equation} 
\label{eq:Jordan_from_3-graded_2}
\{u,v,w\} = [[u,v],w]  
\end{equation}
whenever $u,w \in V_\sigma$ and $v \in V_{-\sigma}$.

Conversely, one can construct a 3-graded Lie algebra $\g$ from a Jordan pair $V$.  Here we take
\begin{equation}
  \g_{-1} =V_{-} , \quad \g_0 = \inn(V), \quad \g_1 = V_{+}
\end{equation}
We define the Lie bracket on $\g$ as follows.  Suppose $v=(v_{+},v_{-})$ and $w=(w_{+},w_{-})$,
with $v_{\sigma}, w_{\sigma}\in V_{\sigma}$.  Then 
\begin{itemize}
\item $[v_\sigma, w_\sigma] = 0$,
\item $[v_+, v_-] = \delta_v$,
\item $[\delta_v, w_\sigma] = \delta_v w_\sigma$, 
\item $[\delta_v, \delta_w]$ is the commutator of inner derivations, as in \eqref{eq:bracket_of_inner_derivations}.
\end{itemize}
This process of going from a Jordan pair to a 3-graded Lie algebra is called the \define{Tits--Kantor--Koecher construction}.

The Tits--Kantor--Koecher construction is not inverse to going from a 3-graded Lie algebra to a Jordan pair.  For example, there is no guarantee that starting from a 3-graded Lie algebra $\g$ and building a Jordan pair $V$, we get $\textrm{Inn}(V) = \g_0$.  Technically, the two processes are adjoint functors, with the functor from 3-graded Lie algebras to Jordan pairs being the right adjoint \cite{CavenySmirnov2014}.

To get a \emph{hermitian} Jordan pair from a 3-graded Lie algebra $\g$, we need to equip $\g$ with some extra structure. Namely, we need a grade-reversing antilinear involution: that is, an antilinear map $\theta \maps \g \to \g$ that preserves the Lie bracket, obeys $\theta^2 = 1$, and has $\theta \maps \g_i \to \g_{-i}$.  From $\g$ we can build a Jordan pair $V$ as before, with $V_+ = \g_1$ and $V_- = g_{-1}$.  Then we can equip $V$ with an antilinear involution $\kappa$ that is simply $\theta$ on each space $V_\pm$.  It is easy to check that $(V,\kappa)$ obeys the axioms of a hermitian Jordan pair.

Conversely, there is a process for getting a 3-graded Lie algebra with grade-reversing antilinear involution from a hermitian Jordan pair $(V,\kappa)$.  We use the previously described method to get a 3-graded Lie algebra $\g$ from $V$.  Then we equip $\g$ with a grade-reversing antilinear involution $\theta$ as follows:
\begin{itemize}
\item $\theta \maps \g_1 \to \g_{-1}$ is $\kappa_+ \maps V_+ \to V_-$,
\item $\theta \maps \g_{-1} \to \g_1$ is $\kappa_- \maps V_- \to V_+$,
\item $\theta \maps \g_0 \to \g_0$ maps each inner derivation $\delta_v$, where $v = (v_+, v_-)$, to the inner derivation $\delta_{v'}$, where $v' = -(\kappa v_-, \kappa v_+)$.
\end{itemize}
In general, this process is not inverse to the process for getting a hermitian Jordan pair from a 3-graded Lie algebra with grade-reversing antilinear involution.  But we shall see that in the main examples of interest they will act as inverses.

Given a complex 3-graded Lie algebra $\g$ with a grade-reversing antilinear involution $\theta$, the fixed points of $\theta$ form a real Lie algebra
\begin{equation}
\label{eq:k}
\k = \{x \in \g \; \vert \; \theta x = x \}
\end{equation}
Moreover $\k$ is naturally a $\Z_2$-graded Lie algebra---not a Lie superalgebra, but an ordinary Lie algebra decomposed into even and odd parts, $\k = \k_0 \oplus \k_1$, with $[\k_i, \k_j] \subseteq \k_{(i+j) \bmod 2}$.  To define the grading we set
\begin{subequations}
  \begin{eqnarray}
    \k_0 &=& \{x \in \g_0 \; \vert \; \theta x = x \}  \label{eq:k_0} \\
    \k_1 &=& \{x \in \g_{-1} \oplus \g_{1} \; \vert \; \theta x = x \}.   \label{eq:k_1}
  \end{eqnarray}  
\end{subequations}

What does all this look like when $\g$ arises from a hermitian Jordan pair $(V,\kappa)$, or equivalently a hermitian Jordan triple $W$?   First, $\k_0$ is the real Lie algebra we called $\inn_\R(W)$ at the end of Section \ref{subsec:derivations}: it is isomorphic to the Lie algebra of \emph{real} inner derivations of $V$, meaning those commuting with $\kappa$.   Second, $\k_1$ is isomorphic, as a vector space, to the underlying real vector space of $\g_1$, since elements of $\k_1$
can all be written uniquely in the form $z + \theta z$ with $z \in \g_1$.   Thus we can identify $\k_1$ with $\g_1 = W$.  

Summarizing, the $\Z_2$-graded real Lie algebra $\k$ associated to a hermitian Jordan triple $W$ is
\begin{equation}
\label{eq:k_summary}
  \k = \inn_{\R}(W) \oplus W
\end{equation}
with a suitably defined Lie bracket.

\subsection{Reinterpretation as symmetric spaces}
\label{subsec:symmetric_spaces}

Positive hermitian Jordan triples have yet another reinterpretation, namely as compact hermitian symmetric spaces.  Hermitian symmetric spaces were introduced and classified by Cartan \cite{Cartan1935} and Borel \cite{Borel1952, borel2006compactifications}.  Later, Loos studied symmetric spaces \cite{Loos1969symmetricVolI, Loos1969symmetricVolII} and developed the connection between hermitian symmetric spaces, hermitian Jordan pairs, and hermitian Jordan triples \cite{Loos1975, Loos1977}.  We sketch only what we need of this theory.

A hermitian manifold is a natural complex geneneralization of a Riemannian manifold: it is a complex manifold equipped with a smoothly varying complex inner product on each tangent space.  A hermitian symmetric space is the corresponding generalization of a Riemannian symmetric space: it is a connected hermitian manifold $M$ that, for each point $p \in M$, has an inversion symmetry $\sigma_p \maps M \to M$ that respects the hermitian structure.  Here by `inversion symmetry' we mean a smooth map that fixes $p$, squares to the identity, and acts as $-1$ on the tangent space of $p$:
\begin{equation}
\sigma_p (p) = p, \qquad \sigma_p^2 = 1, \qquad d\sigma_p(p) = -1.
\end{equation}

Let us see how a positive hermitian Jordan triple $W$ gives rise to a compact hermitian symmetric space.  We have already seen how $W$ gives a 3-graded Lie algebra $\g$ equipped with grade-reversing antilinear involution $\theta$. From this we shall build a homogeneous space, which will turn out to be the desired compact hermitian symmetric space. 

First, let $G$ be the connected and simply-connected Lie group whose Lie algebra is $\g$.  This is always a complex semisimple Lie group.   Second, let $\p \subseteq \g$ be the Lie subalgebra with 
\begin{equation}
\p = \g_{-1} \oplus \g_0 
\end{equation}
as a vector space.  
Let $P$ be the Lie subgroup of $G$ generated by elements of $\p$.  $P$ is a complex Lie group, called a parabolic subgroup of $G$.  

The quotient $G/P$ is not only a complex manifold:  in fact it can be given a Riemannian structure.  To do this we need the involution $\theta$, which we have not used yet. Let $K$ be the Lie subgroup of $G$ generated by the real Lie subalgebra
\begin{equation}
\k = \{x \in \g \; \vert \; \theta x = x \}
\end{equation}
introduced in equation \eqref{eq:k}.  $K$ is a compact Lie group.  Similarly let $K_0$ be the compact Lie subgroup of $G$ generated by the real Lie subalgebra 
\begin{equation}
\k_0 = \{x \in \g_0 \; \vert \; \theta x = x \}.
\end{equation} 
introduced in equation \eqref{eq:k_0}.  We have $\k_0 = \p \cap \k$ and $K_0 = P \cap K$.

An important result states that
\[    G/P \cong K/K_0  .\]
At left we clearly have a complex manifold; at right we have a compact real manifold that can be given a $K$-invariant Riemannian metric by averaging any Riemannian metric over the action of $K$.   Remarkably, these structures fit together in an optimal way, making $G/P \cong K/K_0$ into a compact hermitian symmetric space on which $K$ acts preserving the hermitian structure.  The larger group $G$ does not, though it preserves the complex structure.

Thus, we have a way to construct a compact hermitian symmetric space from a positive hermitian Jordan pair.  Even better, every compact hermitian symmetric space arises from this procedure.  Indeed, Loos showed that there is an equivalence between these two concepts.  Thus, Cartan's classification of compact hermitian symmetric spaces corresponds to a classification of positive hermitian Jordan pairs.  This classification goes as follows.

A compact hermitian symmetric space is said to be \define{irreducible} if it is not a product of other such spaces.   Every compact hermitian symmetric space is covered by a product of irreducible ones.  The irreducible compact hermitian spaces are precisely the homogeneous spaces of the form $M = K/K_0$, where $K$ is a compact connected simple Lie group with trivial center, and $K_0$ is a connected Lie subgroup having the same rank as $K$ and having $\U(1)$ as center \cite[Ch.\ VII.6]{Helgason2001}.

The $\U(1)$ center of $K_0$ has an important meaning.  Any complex manifold $M$ has tangent spaces that are complex vector spaces, so $\U(1)$ acts on each tangent space $T_pM$ as multiplication by phases.  In the hermitian symmetric space $M = K/K_0$, each point $p \in M$ is stablized by a subgroup conjugate to $K_0$, and the $\U(1$) action on the tangent space comes differentiating the action of this $\U(1)$ subgroup of $K_0$ on $M$.   When we act by $-1 \in \U(1)$ we recover the inversion map $\sigma_p$ characteristic of a symmetric space.  

Cartan and Borel classified irreducible compact hermitian symmetric spaces and show they come in four infinite families and two exceptional cases. These correspond to positive hermitian Jordan triples that are simple---that is, not direct sums of others in a nontrivial way.    The simple positive hermitian Jordan triples are listed in Table \ref{tab:hermitian_jordan_triples}, and the corresponding irreducible compact hermitian symmetric spaces are listed in Table \ref{tab:hermitian-symmetric-spaces}.  For both we have used Cartan's numbering system for symmetric spaces.

\begin{table}[h]
  \centering
  \caption{Simple positive hermitian Jordan triples.}
  \label{tab:hermitian_jordan_triples}
  \begin{tabular}{@{}ll@{}}
    \toprule
     Type & Jordan triple \\
    \midrule
    $\mathrm{A\,III}$
      & $\M_{p,q}(\C)$, $p \times q$ complex matrices, \\ & with $[x,y,z] = \tfrac{1}{2}(xy^\dagger z+ zy^\dagger x)$ \\
    $\mathrm{D\,III}$
      & $\mathfrak{a}_n(\C)$, skew-symmetric $n\times n$ complex matrices, \\
      & \quad with $[x,y,z] = \tfrac{1}{2}(xy^\dagger z+ zy^\dagger x)$ \\
    $\mathrm{C\,I}$
      & $\mathfrak{s}_n(\C)$, symmetric $n\times n$ complex matrices, \\
      & \quad with $[x,y,z] = \tfrac{1}{2}(xy^\dagger z+ zy^\dagger x)$ \\
    $\mathrm{B\,D\,I}$
      & $\C^n$ with  $[x,y,z] = ( x \cdot \bar y) z + (z \cdot \bar y) x - ( x, \bar z\rangle \bar y$ \\
    $\mathrm{E\,III}$
      & $\M_{2,1}(\O_\C)$ with $
  [x,y,z]=\frac{1}{2}(x(y^{\dagger}z)+z(y^{\dagger}x))$ \\
    $\mathrm{E\,VII}$
      & $\mathfrak{h}_3(\O_\C)$ with $[x,y,z] = (x \circ y^\ast) \circ z + x \circ (y^\ast \circ z) - y^\ast \circ (a \circ c)$
      \\
    \bottomrule
  \end{tabular}
\end{table}

In the AIII, DIII and CI Jordan triples, $x^\dagger$ denotes the conjugate transpose of a matrix with complex entries.   In BDI, $\bar x$ denotes the componentwise complex conjugate of a vector $x \in \C^n$.   In EIII and EVII, $\O_\C$ denotes the \define{bioctonions} $\C \otimes \O$: the tensor product of the algebras $\C$ and $\O$ over $\R$.  We can apply both complex and octonionic conjugation to a bioctonion.  We use $x^{\ast}$ to denote the entrywise complex conjugate of a matrix $x$ of octonions, $\bar{x}$ to denote its entrywise octonionic conjugate,  $x^{T}$ to denote its transpose, and $x^\dagger$ to denote $(\bar{x}^{\ast})^{T}$.  In EVII, $\h_3(\O_\C)$ is the space of $3 \times 3$ $\O_\C$-valued matrices $x$ with $x = \bar{x}^T$.   

We call $\h_3(\O_\C)$ the \define{Albert triple} and  $\M_{2,1}(\O_\C) \cong \O_\C^2$ the \define{bi-Cayley triple}.  We describe these in detail in Sections \ref{subsubsec:Albert} and \ref{subsubsec:bi-Cayley}, respectively.   These ``exceptional'' hermitian Jordan triples are the basis of our approach to the Standard Model.   Our approach also relies on the fact that the Standard Model gauge group is the group $K_0$ for a particular Jordan triple of type AIII.  The AIII series of Jordan triples, consisting of spaces $\M_{p,q}(\C)$, have
\begin{equation}      K_0 = \S(\U(p) \times \U(q)) = 
\Big\{ g \in \SU(p+q) \; \Big\vert \;  g = 
\left( 
\begin{array}{c c }
h &  0 \\
0    & h'
\end{array}
\right)
 \; \Big\vert \; h \in \U(p), h' \in \U(q) \Big\}.  
\end{equation}
In particular, $\M_{3,2}(\C)$ has 
\begin{equation} 
\label{eq:standard_model_group}
K_0 = \S(U(3) \times \U(2)) \cong \frac{\SU(3) \times \SU(2) \times \U(1)}{\Z_6} 
\end{equation}
for a particular $\Z_6$ subgroup of the center of $\SU(3) \times \SU(2) \times \U(1)$, which acts trivially on all known particles \cite{BaezHuerta:2009}.

\begin{table}[t]
  \centering
  \caption{Irreducible compact hermitian symmetric spaces.}
  \label{tab:hermitian-symmetric-spaces}
  \begin{tabular}{@{}llllp{0.38\textwidth}@{}}
    \toprule
     Type & $G$ & $K$ & $K_0$ & The compact hermitian $\quad$ symmetric space $K/K_0$ \\
    \midrule
    $\mathrm{A\,III}$
      & $\mathrm{SL}(p+q,\C)$
      & $\SU(p+q)$
      & $\mathrm{S}(\U(p)\times\U(q))$
      & Space of complex $p$-dimensional subspaces of $\C^{p+q}$ \\
    $\mathrm{D\,III}$
      & $\SO(2n,\C)$
      & $\SO(2n)$
      & $\U(n)$
      & Space of orthogonal complex structures on $\R^{2n}$ \\
    $\mathrm{C\,I}$
      & $\Sp(2n,\C)$
      & $\Sp(n)$
      & $\U(n)$
      & Space of complex structures on $\H^{n}$ compatible with the $\quad$ inner product \\
    $\mathrm{B\,D\,I}$
      & $\SO(n+2,\C)$
      & $\SO(n+2)$
      & $\SO(n)\times\SO(2)$
      & Space of oriented real $2$-dimensional subspaces of $\R^{n+2}$ \\
     $\mathrm{E\,III}$
      & $\E_6^{\C}$
      & $\E_6$
      & $(\Spin(10)\times\U(1))/\Z_4$
      & $(\C\otimes\O)\P^{2} = \O_{\C}\P^2$ \\ [6pt]
    $\mathrm{E\,VII}$
      & $\E_7^{\C}$
      & $\E_7$
      & $(\E_6\times\U(1))/\Z_2$
      & Space of symmetric subspaces of $(\H\otimes\O)\P^2$ isomorphic to $(\C\otimes\O)\P^{2}$ \\
    \bottomrule
  \end{tabular}
\end{table}

\section{Jordan pair quantum theory}
\label{sec:quantum}

Euclidean Jordan algebras were originally introduced as a way to reformulate standard quantum theory \cite{Jordan:1933vh} and also generalize it, allowing the study of real, complex and quaternionic theories on an equal footing \cite{Baez2012} along with more exotic possibilities involving the exceptional Jordan algebra $\h_3(\O)$ and so-called spin factors. In this paper we have switched our attention to a related but different class of objects---hermitian Jordan pairs, or equivalently hermitian Jordan triples---because of their possible connection to the Standard Model.  Here we explain how these too offer a way to reformulate standard quantum theory, and generalize it in a different direction.

To give a formulation of quantum theory, we need at least to explain:
\begin{enumerate}
    \item What is a state $e$? 
    \item What is an observable ${\cal O}$? 
    \item What are possible outcomes when we measure ${\cal O}$?  
    \item What is the expectation value of the outcome $\langle{\cal O}\rangle$?
    \item How can the state $e$ evolve in time? 
\end{enumerate}

First we recall the framework.  We start with a hermitian Jordan pair $V=(V_{+},V_{-})$ with triple product $\{\;\bigcdot\;,\;\bigcdot\;,\;\bigcdot\;\}$ and antilinear involution $\kappa$.  Let $W = V_+$ be the corresponding hermitian Jordan triple with skew-linear triple product $[\;\bigcdot\;,\;\bigcdot\;,\;\bigcdot\;]$. In Proposition \ref{prop:Inn_R} we saw that for any pair $v=(v_+,v_-)$ with $v_{\sigma}\in V_{\sigma}$, we have a derivation of $V$ that commutes with $\kappa$, given by $D_v = \delta_v - \delta_{\kappa v}$, where $\delta_v$ was the derivation defined by equation \eqref{eq:delta}.   Correspondingly, on $W$, for any two elements $w,w'\in W$, we have $D_{w,w'} \maps W\to W$ given by
\begin{equation}
  D_{w,w'}
  =[w,w',\;\bigcdot\;]
  -[w',w,\;\bigcdot\;].
\end{equation}
The derivations $D_v$ and $D_{w,w'}$ are anti-hermitian with respect to these inner products.  

With these ingredients in place, our proposed answers to the five questions posed above are:
\begin{enumerate}
    \item A pure state is a minimal tripotent as defined in Section \ref{subsec:tripotents}, \emph{i.e.,} an element $e \in W$ satisfying $[e,e,e]=e$ that cannot be written as the sum of two nonzero mutually orthogonal tripotents (or, equivalently, that has Peirce 1-space of dimension 1 in $W$).
    \item The observables correspond to elements of
    $\inn_{\R}(V)$.  To be precise, since any
    $D\in \inn_{\R}(W)$ is a skew-adjoint operator on $W$, the corresponding observable is ${\cal O}=iD$, which is self-adjoint.  
    \item When we measure an observable ${\cal O}$, the possible outcomes of the measurement are given as usual by the eigenvalues of ${\cal O}$ (which are real).
    \item The expected value $\langle{\cal O}\rangle$ of the observable ${\cal O}$ in the state $e$ is given as usual by
    $\langle e \; \vert \; {\cal O}e\rangle$, where the inner product on $W$ is defined as in equation \eqref{eq:inner_product}.  Note that this inner product is antilinear in the second argument, not the first as physicists prefer.  (This is just a difference of convention.)
    \item The time evolution of pure states is determined by a ``Hamiltonian'', which can be any observable ${\cal H} = iD$, where $D \in \inn_\R(V)$.  The time evolution of a pure state $e(t)$ is then given by this version of Schr\"odinger's equation:
    \begin{equation}
      \label{time_evolution}
        \frac{d}{dt} e(t) = - i{\cal H} e(t)
    \end{equation} 
    The solution of this equation is
    \begin{equation} 
    e(t) = \exp(-it {\cal H}) \; e(0) .
    \end{equation}
    For any time $t \in \R$, the operator $\exp(- it{\cal H} )$ is an automorphism of the hermitian Jordan triple $W$.  Since the concept of minimal tripotent is defined in an automorphism-invariant way, it follows that $e(t)$ is a pure state for all $t$ if $e(0)$ is.
\end{enumerate}
This is our proposal for the Jordan-pair/Jordan-triple reformulation of pure state quantum theory.  We leave the reformulation of mixed-state quantum theory to future work.

This formulation recovers ordinary quantum theory with the $n$-dimensional Hilbert space $\C^n$ when applied to the hermitian Jordan pair $V = (V_+,V_-)$ where $V_+=\M_{n,1}(\C)$ and $V_-=\M_{1,n}(\C)$ are given the triple product $\{x,y,z\}=\frac{1}{2}(xyz+zyx)$ and the antilinear map $\kappa v=v^{\dagger}$.  In this case $V_+ \cong \C^n$ is the familiar space of ``kets'', while $V_- \cong (\C^n)^\ast$ is the space of ``bras'', and $\kappa$ is the usual antilinear isomorphism sending any bra $|e \rangle$ to the corresponding ket $\langle e|$.  The hermitian Jordan triple system is $W = \C^n$ with skew-linear triple product $[x,y,z]=\frac{1}{2}(x y^\dagger z+z y^\dagger x)$.   The Lie algebra $\inn_{\R}(W)$ is $\u(n)$, and the corresponding Lie group is the familiar Lie group of all unitary operators on $\C^n$.   The space of observables $\{{\cal O}\}$ is the space of $n\times n$ Hermitian matrices acting on $\C^{n}$.  The rule for the expected value of an observable reduces to the Born rule, and the time evolution equation (\ref{time_evolution}) reduces to the usual Schr\"odinger equation.

But the Jordan pair formulation also generalizes ordinary quantum theory, since there are other hermitian Jordan triples listed in Table \ref{tab:hermitian_jordan_triples}.  In particular, there are two exceptional hermitian Jordan triples connected to the bioctonions.  These exotic quantum systems can have unfamiliar properties: for example they may contain the observable ${\cal O}$, but not its square!   We are interested in these exotic quantum systems because of their connection to the Standard Model, which we explain in Section \ref{sec:standard}.

Even in the case of ordinary quantum theory, the Jordan triple framework sheds a new light on things.  For example, we can express the framework in the language of 3-graded Lie algebras.  In Section \ref{subsec:graded} we saw that any hermitian Jordan pair $V$ gives rise to a 3-graded complex Lie algebra 
\begin{equation}
  \g = V_{-} \oplus \inn(V) \oplus V_{+},
\end{equation}
with an antilinear involution $\kappa$ which swaps $V_{+}$ and $V_{-}$.  One can argue that:
\begin{itemize}
\item
The degree $0$ component $\g_0 = \inn(V)$ is the Lie algebra of physically relevant operators.  These are all real-linear combinations of symmetry generators (elements of $\inn_\R(V)$, which commute with $\kappa$) and observables (elements of $i\,\inn_\R(V)$, which anti-commute with $\kappa$).
\item
The degree $1$ component $\g_1 = V_{+}$ corresponds to the ``kets'', which we write suggestively as $|e\rangle$.
\item
The degree $-1$ component $\g_{-1} = V_-$ corresponds to the ``bras'' $\langle e|$.
\end{itemize}
So, schematically:
\begin{equation}
  \g=\{\text{bras}\}\oplus \{\text{operators}\}\oplus
  \{\text{kets}\}.
\end{equation}

The Lie bracket in this 3-graded algebra is interesting to consider.  The bracket of two elements of degree $0$ is just the usual commutator of operators.  The bracket of an element of degree $0$ and one of degree $1$ is the action of operators on kets.   Similarly the bracket of an element of degree $0$ and one of degree $-1$ is the action of operators on kets.  A less familiar feature is the bracket of a bra and a ket.  This is an operator. 

The Lie algebra $\g$ has a real Lie subalgebra $\k$ consisting of the fixed points under the antilinear involution $\theta \maps \g \to \g$.  We saw in equation \eqref{eq:k_summary} that this is a $\Z_2$-graded algebra, with the Lie algebra of real inner derivations of $W$ as its even part, and $W$ as its odd part:
\begin{equation}
\k = \inn_\R(W) \oplus W 
\end{equation}
Since elements of $\inn_\R(W)$ correspond to observables after multiplication by $i$, we have 
\begin{equation}
  \k \cong \{\text{observables}\}\oplus
  \{\text{states}\}.
\end{equation}
The inclusion of $\k$ in $\g$ sends observables ${\cal O}$ to skew-adjoint operators $- i {\cal O}$ and states $e$ to bra-ket pairs $(\langle e| , |e \rangle)$.

How does this look in the case of ordinary quantum theory where the space of states is $\C^n$?   Here we get $G =\SL(n+1,\C)$, with its Lie algebra broken into three summands---operators, bras and kets:
\begin{equation}
  \g = \mathfrak{sl}(n+1,\C) =
  \begin{pmatrix}
    \mathfrak{gl}(n,\C) & \C^n \\
    (\C^n)^\ast & *
  \end{pmatrix}
\end{equation}
where the lower right corner is a scalar equal to minus the trace of the upper left matrix in $\gl(n,\C)$.  Similarly we have $K = \SU(n+1)$, with its Lie algebra broken into two summands---observables and states:
\begin{equation} 
\k = \mathfrak{su}(n+1,\C) =
\Big\{ \begin{pmatrix}
    {\cal O} & e \\
    e^\dagger & -\tr({\cal O}) 
  \end{pmatrix} \Big\vert \; {\cal O} \in \u(n), \; e \in \C^n  \Big\}
\end{equation}
Since $\k_0 = \u(n)$, $K_0$ is the familiar unitary group $\U(n)$.  The compact hermitian space $K/K_0$ is $\C\P^n$.  The tangent space at any point of $\C\P^n$ is a copy of the familiar Hilbert space of states, $\C^n$.  Thus, positive hermitian Jordan triples provide a new and somewhat peculiar outlook on ordinary quantum theory. 

\section{The exceptional Jordan pairs and triples}
\label{sec:exceptional}

We have seen that in the classification of hermitian Jordan triples there are two exceptional cases: the Albert triple and the bi-Cayley triple.  Here we introduce these in more detail, as a warmup to discussing their relation to the Standard Model.   None of the material in this section is new, though the details are scattered in the literature and somewhat hard to find \cite{Chu2012, Freudenthal1964,Gunaydin1993,Baez:2001dm,CorradettiMarraniChesterAschheim2022Conjugation}.

\subsection{The octonions and bioctonions}

A celebrated theorem due to Hurwitz states that there are only four normed division algebras: the real numbers $\R$, the complex numbers $\C$, the quaternions $\H$, and the octonions $\O$, with dimensions 1, 2, 4 and 8, respectively \cite{Baez:2001dm}.  Just as a complex number $z\in\C$ may be written $z=a_{0}+a_{1}i$ and a quaternion $q\in\H$ may be written $q=a_{0}+q_{1}i+a_{2}j+a_{3}k$, an octonion $x\in\O$ may be written 
\begin{equation}
  \label{octonion}
  x=a_{0}+a_{1}i+a_{2}j+a_{3}k+a_{4}l+a_{5}\;\!il+a_{6}\;\!jl+a_{7}\;\!kl
\end{equation}
where the coefficients $a_i$ ($i=0,\ldots,7$) are real numbers, each of the seven imaginary units $\{i,j,k,l,il,jl,kl\}$ squares to $-1$, and the product of any two {\it distinct} imaginary units gives a third imaginary unit, according to the rule shown by the so-called ``Fano plane'' in Fig.~\ref{FanoPlane}.  

Any octonion $x$ has octonionic conjugate $\bar{x}$ given by negating the $a_{i}$ for $i=1,\ldots,7$, and 
\[       x\bar{x} = \bar{x} x \ge 0 \]
with equality only for $x = 0$.  If we define $\|x\| = \sqrt{x\bar{x}}$ then $\|xy\| = \|x\| \, \|y\|$.  Thus, we say the octonions are a normed division algebra.

\begin{figure}
  \begin{center}
    \includegraphics[width=1.9in]{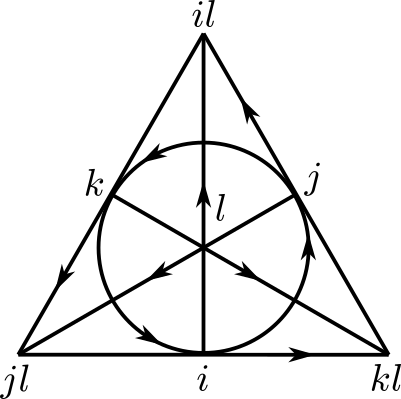}
  \end{center}
  \caption{Fano plane summarizing octonionic multiplication, and our conventions.  See Sec. 2.1 in \cite{Baez:2001dm} for more explanation.}
 \label{FanoPlane}
\end{figure}

More generally, we can consider $\O_{K}$, the octonion algebra over a field $K$.  This is defined just as before, except now the coefficients $a_i$ are valued in $K$.   For the application to the Standard Model discussed in this paper, we are interested in the case $K=\C$ --- {\it i.e.}\ the {\it complex} octonions $\O_\C = \C \otimes \O$, also known as the bioctonions.

Like the real octonions $\O=\O_{\R}$, the bioctonions $\O_\C$ are special in many ways: for example, $\R_\C$, $\C_\C$, $\H_\C$ and $\O_\C$ 
are the only four {\it composition} algebras over $K=\C$, and also the only four {\it alternative} algebras over $K=\C$.  Here an algebra $A$ over a field $K$ is a \define{composition algebra} if it has a quadratic form $N \maps A\to K$ satisfying $N(xy)=N(x)N(y)$, and an \define{alternative algebra} if its associator $(a,b,c)=(ab)c-a(bc)$ is completely antisymmetric in its three arguments.  Equivalently, an algebra is alternative if the subalgebra generated by any two elements is associative.

In many different parts of mathematics, when interesting mathematical objects are classified, they fall into regular families, plus a few exceptional cases.  These exceptional cases tend to be related to the octonions, and to one another \cite{Baez:2001dm}. The two exceptional positive hermitian Jordan pairs, or positive hermitian Jordan triples, are examples.




\subsection{The exceptional hermitian
Jordan pairs and triples}

As seen in Table \ref{tab:hermitian_jordan_triples}, the classification of positive hermitian Jordan triples yields two special cases: the Albert triple and the bi-Cayley triple, which is in fact a subtriple of the Albert triple.  We introduce these in more detail here, along with their corresponding hermitian Jordan pairs.

\vskip 1em

\textbf{Notation.} Given a matrix $x$ of bioctonions, we use $x^{\ast}$ to denote its entrywise complex conjugate, $\bar{x}$ to denote its entrywise octonionic conjugate, $x^{T}$ to denote its transpose, and $x^{\dagger}$ to denote $(\bar{x}^{\ast})^{T}$, where we transpose the matrix $x$ after taking the octonionic \emph{and} complex conjugate of each entry.

\subsubsection{The Albert pair and triple}
\label{subsubsec:Albert}

Let $\h_3(\O_\C)$ denote the \define{Albert algebra}: the 27-dimensional complex vector space of $3\times3$ matrices $x$ with entries in $\O_\C$ such that $x = \bar{x}^T$ (that is, $x_{ij} = \bar{x}_{ji}$).  Let $\circ$ denote the following symmetric bilinear product on this vector space:
\begin{equation}
    x\circ y=\frac{1}{2}(xy+yz)
\end{equation}
where $xy$ denotes ordinary matrix multiplication.  This product $\circ$ makes $\h_3(\O_\C)$ into a Jordan algebra over the complex numbers.

The \define{Albert pair} is the Jordan pair
\begin{equation}
    V=(V_+,V_-)=(\h_3(\O_\C),
    \h_3(\O_\C))
\end{equation}
with the complex-trilinear triple product
\begin{subequations}
  \begin{eqnarray}
  \{x,y,z\}&=&x\circ(y\circ z)+(x\circ y)\circ z
  -y\circ(x\circ z) \\
  &=&\frac{1}{2}(x(yz)+z(yx)).
  \end{eqnarray}
\end{subequations}
The inner derivations of the Albert pair form the complex Lie algebra 
\begin{equation}
  \inn(V) = \e_6^{\C} \oplus \C
\end{equation}
and the Tits--Kantor--Koecher construction applied to the Albert pair gives the complex Lie algebra $\e_7^{\C}$ with this 3-grading:
\begin{subequations}
  \begin{eqnarray}
  \e_7^{\C} &=&V_- \oplus \inn(V)\oplus V_+ \\ [3pt]
  &=& \h_3(\O_\C)\oplus 
  (\e_{6}^{\C}\oplus\C) \oplus \h_3(\O_\C).
\end{eqnarray}
\end{subequations}

The antilinear involution $\kappa \maps V\to V$ ($\kappa_{\sigma} \maps V_{\sigma}\to V_{-\sigma}$) 
given by 
\begin{equation}
  \kappa x = x^{\dagger}
\end{equation}
makes the Albert pair into a hermitian Jordan pair.  
The {\bf Albert triple} is the corresponding  hermitian Jordan triple
\begin{equation}
  W = \h_3(\O_\C)
\end{equation}
with skew-linear triple product 
\begin{subequations}
  \begin{eqnarray}
    [x,y,z]&=&x\circ(y^{\ast}\circ z)+(x\circ y^{\ast})\circ z
  -y^{\ast}\circ(x\circ z) \\
  &=&\frac{1}{2}(x(y^{\dagger}z)+z(y^{\dagger}x)).
  \end{eqnarray}
\end{subequations}
The inner derivations of this hermitian Jordan triple form a real Lie algebra
\begin{equation}
  \inn_{\R}(W) = \e_{6}^{\R}\oplus \u(1).
\end{equation}
where $\e_6^\R$ denotes the compact real form of $\e_6$. 

The antilinear involution $\kappa$ defines a grade-reversing antilinear involution on $\e_7^{\C}$, whose fixed points form a $\Z_2$-graded real Lie algebra as in equation \eqref{eq:k}.  This turns out to be the compact real form of $\e_7$:
\begin{equation}
    \inn_\R(W)\oplus W =
    (e_{6}^\R \oplus \u(1)) \oplus h_3(\O_\C) = \e_7^\R
\end{equation}
whose odd part $\h_3(\O_\C)$ transforms as the 
complex 27-dimensional fundamental representation of 
$\e_{6}^{\R}$. 

The Albert triple also gives rise to the hermitian symmetric space that Cartan called EVII (see Table \ref{tab:hermitian-symmetric-spaces}).  This has complex dimension 27, and it is the space of all symmetric subspaces of the Rosenfeld plane $(\H\otimes\O)\P^2$ that are isomorphic to the Rosenfeld plane $(\C \otimes \O)\P^2$.  In the context of this symmetric space, the Lie algebra $\e_{7}^{\R}$ has the following geometric interpretation.  The full Lie algebra $\e_{7}^{\R}$ generates the isometry group of this symmetric space.  This is a $\Z_2$-graded Lie algebra whose even part $\e_{6}^\R \oplus \u(1)$ generates the subgroup preserving a given point $p$ of this symmetric space, while its odd part $\h_3(\O_\C)$ describes the tangent space at $p$.  The Lie subalgebra $\u(1)$ generates a $\U(1)$ action which simply acts by complex phases on the tangent space at $p$.

\subsubsection{The bi-Cayley pair and triple}
\label{subsubsec:bi-Cayley}

Let $M_{p,q}(\O_\C)$ denote the space of $p\times q$ matrices with entries in $\O_\C$.  The \define{bi-Cayley pair} is the Jordan pair
\begin{equation}
    V=(V_+,V_-)=(M_{1,2}(\O_\C),M_{2,1}(\O_\C))
\end{equation}
with trilinear product
\begin{equation}
    \{x,y,z\}=\frac{1}{2}(x(yz)+z(yx)).
\end{equation}
The inner derivations of this pair form the complex Lie algebra 
\begin{equation}
     \inn(V)= \so(10,\C) \oplus \C
\end{equation}
and the Tits--Kantor--Koecher construction applied to the bi-Cayley pair gives the complex Lie algebra $\e_6^{\C}$ with this 3-grading:
\begin{subequations}
\begin{eqnarray}
    \e_6^{\C}&=& V_-\oplus\inn(V)\oplus V_+ \\
    &=& M_{1,2}(\O_\C)\oplus
    {[\so(10,\C)+\C]}\oplus M_{2,1}(\O_\C).
\end{eqnarray}
\end{subequations}

The antilinear involution 
$\kappa:V\to V$ ($\kappa_{\sigma}:V_{\sigma}\to V_{-\sigma}$) given by 
\begin{equation}
    \kappa x = x^{\dagger}
\end{equation}
makes the bi-Cayley pair into a hermitian Jordan pair. The \define{bi-Cayley triple} is the corresponding hermitian Jordan triple
\begin{equation}
    W=\M_{2,1}(\O_\C)
\end{equation}
with skew-linear triple product
\begin{equation}
  [x,y,z]=\frac{1}{2}(x(y^{\dagger}z)+z(y^{\dagger}x)).
\end{equation}
The inner derivations of this triple form the compact real Lie algebra
\begin{equation}
  \inn_{\R}(W)= \so(10,\R)\oplus \u(1).
\end{equation}

The antilinear involution $\kappa$ defines a grade-reversing antilinear involution on $\e_6^{\C}$, whose fixed points form a $\Z_2$-graded real Lie algebra as usual.  This is the compact real form of $\e_6$:
\begin{equation}
    \e_6^{\R} = \inn_{\R}(W) \oplus W=
    (\so(10,\R)\oplus \u(1))\oplus \O_{\C}^2.
\end{equation}

The bi-Cayley triple gives rise to the hermitian symmetric space that Cartan called EIII (again see Table \ref{tab:hermitian-symmetric-spaces}), also called the bioctonionic Rosenfeld plane, $(\C \otimes \O)P^2$.  In the context of this symmetric space, the Lie algebra $\e_{6}^{\R}$ has the following interpretation.  The full Lie algebra $\e_{6}^{\R}$ generates the isometry group of $(\C \otimes \O)\P^2$.  This is a $\Z_2$-graded Lie algebra whose even part $\so(10) \oplus \u(1)$ generates the subgroup preserving a given point $p$ of this symmetric space, while its odd part $\O_{\C}^2$ describes the tangent space at $p$.  As before, the Lie algebra $\u(1)$ generates a $\U(1)$ action which acts as multiplication by phases on the tangent space at $p$.

\subsubsection{Relation of the Albert and bi-Cayley pairs and triples}
\label{subsubsec:Albert_and_bi-Cayley}

The bi-Cayley triple is a ``subtriple'' of the Albert triple: any element of the bi-Cayley triple, say
\[
  \begin{pmatrix}
    x  \\
    y
  \end{pmatrix}  \in \O_C^2 ,
\]
maps to an element of the Albert triple, namely
\[
  \begin{pmatrix}
    0 & 0 & x \\
    0 & 0 & y \\
    \bar x & \bar y & 0 
  \end{pmatrix} \in \h_3(\O_C).
\]
This map is an injective homomorphism of hermitian Jordan triples.  There is an analogous story for the corresponding hermitian Jordan pairs.   

Indeed, this is an example of a more general fact: if we pick any minimal tripotent in the Albert triple, its Peirce $1/2$-space is a copy of the bi-Cayley triple.   



\section{Relation to the Standard Model}
\label{sec:standard}

The Standard Model has gauge group
\begin{equation}
    G_\SM=(\SU(3)\times \SU(2)\times \U(1))/\Z_{6}
\end{equation} 
and the Standard Model fermions come in three generations, each of which transforms under this gauge group in the 16-dimensional complex representation
\begin{eqnarray}
  \label{rho_SM}
  \rho_\SM\!&\!=\!&\!
  \color{cyan}(3,2,\!+\frac{1}{6})
  \color{black}\oplus
  \color{magenta}(\bar{3},1,\!+\frac{1}{3})
  \color{black}\oplus
  \color{darkorange}(\bar{3},1,\!-\frac{2}{3}) \oplus
  \color{darkred}(1,2,\!-\frac{1}{2})
  \color{black}\oplus
  \color{blue}(1,1,\!+\;\!1\;\!)
  \color{black}\oplus
  (\;\!1,1,0), \nonumber\\
  \!&\!=\!&\! 
  \color{cyan}q_{L}
  \color{black}\oplus 
  \color{magenta}\bar{d}_{R}
  \color{black}\oplus 
  \color{darkorange}\bar{u}_{R} 
  \oplus 
  \color{darkred}\ell_{L}
  \color{black}\oplus
  \color{blue}\bar{e}_{R}
  \color{black}\oplus 
  \bar{\nu}_{R} 
\end{eqnarray}
where, in the second line, we have added the particle names physicists attach to these 6 irreducible representations.

In this section we explain how $G_\SM$ and $\rho_\SM$ are obtained from the two exceptional hermitian Jordan triples (Albert and bi-Cayley) by:
\begin{itemize}
    \item either starting from the smaller exceptional hermitian Jordan triple---the bi-Cayley triple $\O_{\C}^{2}$---and picking {\it two} minimal tripotents ({\it i.e.}\ two pure quantum states) that are \define{colinear}: that is, in each other's Peirce $1/2$-space,
    \item or starting from the larger exceptional triple---the Albert triple $\mathfrak{h}_{3}(\O_{\C})$---and picking {\it three} minimal tripotents that are mutually colinear. 
\end{itemize}
In both cases, we will refer to Table \ref{tab:SM_triples}.

\begin{table}[ht]
\centering
\begin{tabular}{@{}lll@{}}
\toprule
Jordan triple & $\Z_{2}$-graded Lie algebra & Symmetric space \\
\midrule
$\tilde{W} = \mathfrak{h}_3(\O_\C)$
  & $\e_{7}=[\e_{6}\oplus \u(1)]\oplus \mathfrak{h}_{3}(\O_\C)$
  & $\mathrm{E}_7\,/\,[(\E_6 \times\U(1))/\Z_2]$ \\[3pt]
$W = \O_\C^2$
  & $\e_{6}=[\so(10)\oplus \u(1)]\oplus\O_\C^2$
  & $\mathrm{E}_6\,/\,[(\Spin(10) \times\U(1))/\Z_4]$ \\[3pt]
$W' = \mathfrak{a}_{5}(\C)$
  & $\so(10)=[\su(5)\oplus \u(1)]\oplus \mathfrak{a}_{5}(\C)$
  & $\Spin(10)\,/\,[\U(5)]$ \\[3pt]
$W'' = \M_{3,2}(\C)$
  & $\su(5)=[\g_{\SM}]\oplus \M_{3,2}(\C)$
  & $\SU(5)\,/\,G_\SM$ \\
\bottomrule
\end{tabular}
\caption{The sequence of Jordan triples leading to the Standard Model. From left to right in each line we have a positive hermitian Jordan triple, the corresponding real $\Z_{2}$-graded Lie algebra with bosonic part in square brackets, and the corresponding hermitian symmetric space.  If we pick any mininal idempotent in a given triple, its Peirce $1/2$-space is the triple in the row below. In the bottom row, $\g_{\SM}=\su(3)\oplus \su(2)\oplus \u(1)$ and $G_\SM=(\SU(3)\times \SU(2)\times \U(1))/\Z_{6}$.}
\label{tab:SM_triples}
\end{table}


\subsection{The Standard Model from the bi-Cayley triple $\O_{\C}^2$}
\label{SM_from_bi-Cayley}

Here we explain how the Standard Model gauge group and its representation on one generation of fermions neatly emerge from the bi-Cayley triple $W=\O_{\C}^{2}$.

First, if we pick two minimal tripotents ({\it i.e.} two pure quantum states) $e_{1}$ and $e_{2}$ that are colinear in $W$, then the subspace $W''\subset W$ that is colinear with $e_{1}$ and $e_{2}$ is the subtriple $W''=\M_{3,2}(\C)$ whose inner automorphism group is precisely the Standard Model gauge group $G_\SM=[\SU(3)\times \SU(2)\times \U(1)]/\Z_{6}$.  

Moreover, picking the two colinear minimal tripotents $e_1, e_2\in W$ selects a natural embedding of $G_\SM$ in $\Inn_{\R}(W)=(\Spin(10) \times \U(1))/\mathbb{Z}_{4}$ under which the original bi-Cayley triple $W=\O_{\C}^{2}$ naturally transforms precisely as a generation of Standard Model fermions (with representation $\rho_\SM$ of $G_{SM}$)---with the six irreducible components of $\rho_\SM$ (the six particle types in a single generation) precisely corresponding to the six components of the Peirce decomposition of $W$ with respect to $e_1$ and $e_2$.

We explain this in several steps.

\subsubsection{The initial triple $W=\O_{\C}^{2}$ resembles an $\SO(10)$ GUT}

First, note from the second row of Table \ref{tab:SM_triples} that the bi-Cayley triple $W=\O_{\C}^{2}$ has inner automorphism group $(\Spin(10) \times \U(1))/\Z_{4}$, and gives rise to the $\Z_{2}$-graded Lie algebra $\Inn_{\R}(W)\oplus W$, whose even or ``bosonic'' part $\Inn_{\R}(W)$ and odd or ``fermionic'' part $W$ transform as the adjoint representation and the complex 16-dimensional spinor representation of $\Spin(10)$, respectively.  

Ignoring for a moment the extra factor of $\U(1)$, this initial structure bears a notable resemblance to the structure of $\SO(10)$ grand unification, where the gauge group is really $\Spin(10)$, and the gauge boson fields transform in the adjoint representation of $\Spin(10)$, while the fermionic fields of a single generation transform under the complex 16-dimensional spinor representation, $\S_{16}$.

\subsubsection{The Standard Model gauge group $G_\SM$}
\label{subsubsec:group}

Next, note that if we pick any two minimal tripotents ({\it i.e.} two pure quantum states) $e_{1}, e_{2}\in W$ which are colinear ({\it i.e.}\ in each other's Peirce $1/2$-space), then the subspace of $W$ that is colinear with both $e_1$ and $e_2$ is the Jordan triple $\M_{3,2}(\C)$.  This has real inner automorphism group $G_\SM$.

In Table \ref{tab:SM_triples} we outline how to reach this result in iterated stages:
\begin{itemize}
\item
We start with the bi-Cayley triple $W=\O_{\C}^2$ on the second row of the table, with real inner automorphism group \begin{equation}
K_0 = \Inn_\R(W)=(\Spin(10) \times \U(1))/\mathbb{Z}_{4}.
\end{equation}
\item
If we then pick any minimal idempotent $e_1\in W$, its Peirce $1/2$-space $W'\subset W$ is the Jordan subtriple $W'=\mathfrak{a}_{5}(\C)$ on the third row of the table.   Recall that $\mathfrak{a}_5(\C)$ is the hermitian Jordan triple consisting of antisymmetric $5 \times 5$ complex matrices, with real inner automorphism group 
\begin{equation}
K_{0}'=\Inn_{\R}(W')=\S(\U(5) \times \U(1)) \cong \SU(5)\times \U(1).
\end{equation}
\item
If we then pick any minimal tripotent $e_2\in W'$, its Peirce $1/2$-space $W''\subset W'$ is the Jordan subtriple $W''=\M_{3,2}(\C)$ on the fourth row, with inner automorphism group described in equation \eqref{eq:standard_model_group}:
\begin{equation}
  K_{0}''=\Inn_{\R}(W'')= (\SU(3)\times \SU(2)\times \U(1))/\Z_{6}=G_\SM
\end{equation}
\end{itemize}
Thus, we have reached the Standard Model gauge group.


\subsubsection{The Standard Model representation $\rho_\SM$}
\label{subsubsec:rep}

Finally, we show that picking two colinear minimal tripotents $e_1, e_2\in W$ also selects a particular embedding of $G_\SM$ in $\Inn_{\R}(W)=(\Spin(10) \times \U(1))/\mathbb{Z}_{4}$, with the property that the original bi-Cayley triple $W=\O_{\C}^{2}$ naturally transforms under $G_\SM$ precisely as one generation of Standard Model fermions.

We do this by iterating another procedure that takes us down the rows of Table \ref{tab:SM_triples}:
\begin{enumerate}
\item 
We start with $W = \O_\C^2$ and its real inner automorphism group, which we now call $G$ (against our previous convention for symmetric spaces).
\item
We then pick a minimal tripotent $e_1\in W$ and take the subgroup $G'\subset G$ that: (i) acts with determinant $1$ on $W$ and (ii) preserves $e_1$ up to a phase (and hence preserves its Peirce $1/2$-space $W'\subset W$).  
\item
We then pick a minimal tripotent $e_2\in W'$, and take the subgroup $G''\subset G'$ that: (i) acts with determinant $1$ on $W'$ and (ii) preserves $e_{2}$ up to a phase (and hence preserves its Peirce $1/2$-space $W''\subset W'$). 
\end{enumerate}
In this way, we see that that chain of triples $W''\subset W'\subset W$ leads to a chain of subgroups $G'' \subset G' \subset G$.  In Appendix \ref{sec:appendix} we show that
\begin{equation}
\begin{array}{ccl}
G &\cong& (\Spin(10) \times \U(1))/\Z_2, \\
G' &\cong& \U(5), \\
G'' &\cong& G_\SM.
\end{array}
\end{equation}
In this way we obtain a chain of embeddings $G_\SM \subset \SU(5) \subset \Spin(10)$. It is well known \cite{BaezHuerta:2009,Slansky:1981yr} that if we restrict the 16-dimensional complex spinor representation $\S_{16}$ of $\Spin(10)$ to $\G_\SM$ along this chain, we obtain the Standard representation $\rho_\SM$.  

A further striking consequence is that the six types of particle in one generation of fermions, listed in equation \eqref{rho_SM}, correspond to the six components of the Peirce decomposition of $\O_{\C}^{2}$ with respect to the tripotents $e_1$ and $e_2$.

We can check this by writing a general element $w\in W=\O_{\C}^{2}$ as $w=(\alpha,\beta)$ where 
$\alpha,\beta\in\O_{\C}$ are complex octonions,
and (without loss of generality) taking our two colinear primitive tripotents $e_1,e_2\in W$ to be $e_{1}=(e,0)$ and $e_2=(0,e)$, where $e=\frac{1}{2}(1+i\ell)\in \O_{\C}$ (with $i\in\C$ and $\ell\in\O$) is an idempotent in $\O_{\C}$ and $f$ is its complementary idempotent $f=\frac{1}{2}(1-i\ell)\in\O_{\C}$.  We have
\begin{equation}
e^2=e, \quad f^2=f, \quad ef=fe=0, \quad e+f=1.
\end{equation}
In fact, these two colinear idempotents are the same as the two complex structures on $\O_{\C}$ described by Krasnov \cite{Krasnov:2022meo}.

Using the Peirce projection operators (\ref{Peirce_projectors}) in $W$, along with the 
Peirce decomposition $\alpha=e\alpha e+e\alpha f+f\alpha e+f\alpha f$ of any element $\alpha\in\O_{\C}$, it is straightforward to check that the Peirce projectors act individually as
\begin{eqnarray}
        P_1(e_1)(\alpha,\beta)&=&(e\alpha e,0), \nonumber\\
        P_{1/2}(e_1)(\alpha,\beta)&=&(e\alpha f+f\alpha e,e\beta), \nonumber\\
        P_0(e_1)(\alpha,\beta)&=&(f\alpha f,f\beta), \nonumber\\
        P_1(e_2)(\alpha,\beta)&=&(0,e\beta e), \nonumber\\
        P_{1/2}(e_2)(\alpha,\beta)&=&(e\alpha,e\beta f+f\beta e), \nonumber\\
        P_0(e_2)(\alpha,\beta)&=&(f\alpha,f\beta f).
\end{eqnarray}
and hence act in combination as
\begin{subequations}
\begin{eqnarray}
  P_{\;\!1\;\!}(e_{2})P_{\;\!1\;\!}(e_{1})(\alpha,\beta)&=&(0,0) \nonumber\\
  P_{\;\!1\;\!}(e_{2})P_{\frac{1}{2}}(e_{1})(\alpha,\beta)&=&(0,e\beta e)
  \;\;\;\;\;\;\;\to\;\;
  \color{blue}(1,1,+1)=\bar{e}_{R} \nonumber\\
  P_{\;\!1\;\!}(e_{2})P_{\;\!0\;\!}(e_{1})(\alpha,\beta)&=&(0,0)
\end{eqnarray}
\begin{eqnarray}
  P_{\frac{1}{2}}(e_{2})P_{\;\!1\;\!}(e_{1})(\alpha,\beta)&=&(e\alpha e,0)
  \;\;\;\;\;\;\;\to\;\;(1,1,\;\,0\;\;\!)=\bar{\nu}_{R} \nonumber\\
  P_{\frac{1}{2}}(e_{2})P_{\frac{1}{2}}(e_{1})(\alpha,\beta)&=&(e\alpha f,e\beta f)\;\;\;\!\to\;\;
  \color{cyan}(3,2,+\frac{1}{6})=q_{L}\nonumber\\
  P_{\frac{1}{2}}(e_{2})P_{\;\!0\;\!}(e_{1})(\alpha,\beta)&=&(0,f\beta e)
  \;\;\;\;\;\;\;\!\to\;\;
  \color{magenta}(\bar{3},1,+\frac{1}{3})=\bar{d}_{R}
\end{eqnarray}
\begin{eqnarray}
  P_{\;\!0\;\!}(e_{2})P_{\;\!1\;\!}(e_{1})(\alpha,\beta)&=&(0,0) \nonumber\\
  P_{\;\!0\;\!}(e_{2})P_{\frac{1}{2}}(e_{1})(\alpha,\beta)&=&(f\alpha e,0)
  \;\;\;\;\;\;\to\;\;
  \color{darkorange}(\bar{3},1,-\frac{2}{3})=\bar{u}_{R} \nonumber\\
  P_{\;\!0\;\!}(e_{2})P_{\;\!0\;\!}(e_{1})(\alpha,\beta)&=&(f\alpha f,f\beta f)\;\to\;\;\color{darkred}(1,2,-\frac{1}{2})=\ell_{L}.
\end{eqnarray}
\end{subequations}
where on the right-hand side we have shown how the non-vanishing subspaces transform under the embedding of $G_\SM\in \Inn_{\R}(W)$ explained above.

Thus, the non-vanishing Peirce subspaces of $W$ associated with  the two colinear primitive tripotents $e_1$ and $e_2$ precisely correspond to the 6 irreducible representations in $\rho_\SM$---that is, to the six particle types appearing in a single generation of Standard Model fermions!

\subsection{The Standard Model from the Albert triple $\mathfrak{h}_{3}(\O_{\C})$}

In Subsection \ref{SM_from_bi-Cayley} we started from the {\it smaller} exceptional hermitian Jordan triple, the bi-Cayley triple $W = \O_\C^2$.  We saw that if we choose {\it two} minimal tripotents ({\it i.e.}\ two pure quantum states) in $W$ that are colinear, then their common Peirce $1/2$-space in $W$ is a hermitian Jordan triple $W'' \cong \M_{3,2}(\C)$ whose inner automorphism group is precisely the Standard Model gauge group $G_\SM$.  We also used this to get $G_\SM$ to act on $\O_\C^2$ exactly as it does on one generation of fermions via the representation $\rho_\SM$.

It is natural to ask what happens if we take one step further back and start instead from the {\it larger} exceptional hermitian Jordan triple: the Albert triple $\tilde{W} = \h_{3}(\C)$.   

In this case, if we choose {\it three} minimal tripotents $e_{0},e_{1},e_{2}\in \tilde{W}$ that are colinear, then extending the analysis of Sec.~\ref{subsubsec:group} by one step, we can see their common Peirce $1/2$-space in $\tilde{W}$ is the same hermitian Jordan triple $W''$ whose automorphism group is $G_\SM$.  If one repeats the iterative construction of Section \ref{subsubsec:rep}, now starting at $\tilde{W}$, one finds that the chain of triples $W''\subset W'\subset W\subset \tilde{W}$ leads to this chain of subgroups:
\begin{equation}
\G_\SM \subset \U(5) \subset (\Spin(10) \times \U(1))/\Z_4 \subset (\E_6 \times \U(1))/\Z_2.
\end{equation}

To see this, recall from Section \ref{subsubsec:Albert_and_bi-Cayley} that the bi-Cayley triple is the Peirce $1/2$-space of any minimal tripotent $e_{0}$ in the Albert triple.  This takes us down from the first line of Table \ref{tab:SM_triples} to the second.  We can then continue moving down the table as in the previous section, taking the common Peirce $1/2$-space of more and more tripotents, until we obtain $W''$ as the common Peirce $1/2$-space of $e_0, e_1$ and $e_2$.

An interesting feature of this approach is that it gives an alternative view of $G_{\SM}$ as precisely the subgroup of $\Inn_\R(\tilde{W}) = (\E_6 \times \U(1))/\Z_2$ that preserves the determinant of elements in $\h_{3}(\O_\C)$ and also preserves the first two of these tripotents ($e_0$ and $e_1$) ``on the nose" and the third ($e_2$) up to phase.

However, notice that in this approach, we still get one generation of Standard Model fermions from the action of $\G_\SM$ on the bi-Cayley triple. Does the larger Albert triple have anything interesting to say about particles in the Standard Model, or yet-unobserved particles?  We do not yet know the answer, but the following point seems worth noting.  Once we choose the first minimal tripotent $e_0\in\tilde{W}= \h_{3}(\O_{\C})$, its Peirce 1/2-space is the bi-Cayley triple $W=\O_{C}^{2}$, which then neatly encodes a single generation of standard model fermions, as described in Section \ref{SM_from_bi-Cayley}.  But all three minimal colinear tripotents $\{e_0,e_1,e_2\}$ are equivalent in $h_{3}(\O_{\C})$, so we could equally well choose $e_1$ or $e_2$ as our ``first'' choice, and thereby descended to the bi-Cayley triple in three different ways (related by symmetry). 

It is natural to wonder whether this threefold choice, which may be related to the phenomenon of triality, is related to the existence of three generations of fermions in the Standard Model.  That triality may somehow be related to the existence of three generations is a longstanding hope \cite{Ramond,FureyHughes2025}.  We do not yet see how to incorporate it in our framework, but feel it is worth pointing out as an interesting topic for future research.

\section{Acknowledgments}

Endre Bokor’s work was supported by the UKRI Centre for Doctoral Training in Algebra, Geometry and Quantum Fields (AGQ), Grant Number EP/Y035232/1.  LB is supported by the STFC Consolidated Grant `Particle Physics at the Higgs Centre.' Research at Perimeter Institute is supported by the Government of Canada, through Innovation, Science and Economic Development, Canada and the Province of Ontario via the Ministry of Research, Innovation and Science. 

\appendix
\section{Proof of the claimed embedding of $\G_\SM$}
\label{sec:appendix}

In this appendix we prove a result claimed in Section \ref{subsubsec:rep}, which constructs the Standard Model gauge group as a subgroup of the real inner automorphism group of the bi-Cayley triple.

\begin{thm}
Let $W = \O_\C^2$ be the bi-Cayley triple, and choose two colinear minimal tripotents $e_1, e_2 \in W$. Then there is a subgroup $G''$ isomorphic to $ G_\SM$ of $\Inn_\R(W)=(\Spin(10)\times \U(1))/\Z_{4}$ that satisfies the following: its natural action on $W$ preserves $e_1$ and $e_2$ up to phase, and its natural actions on $W$ and on $W'=P_{1/2}(e_1)W \cong \mathfrak{a}_{5}(\C)$ have determinant $1$. 
\end{thm}

\begin{proof}
To obtain such a subgroup $G''$ of $\Inn_\R(W)$ we show that there is a group isomorphism $\Phi \maps \Inn_\R(W'') \rightarrow G''\subset \Inn_\R(W)$ from the inner automorphism group $\Inn_\R(W'') \cong G_\SM$ of $W''=P_{1/2}(e_1)P_{1/2}(e_2)W \cong \M_{3,2}(\C)$. We show this by finding two injective homomorphisms
\[   \Phi_{e_1} \maps \Inn_\R(W') \to \Inn_\R(W) \]
and 
\[   \Phi_{e_2} \maps \Inn_\R(W'')  \to \Inn_\R(W') \]
such that $\Phi=\Phi_{e_1}\circ \Phi_{e_2}$.
By Table \ref{tab:hermitian_jordan_triples} we have $\Inn_\R(W') \cong \U(5)$ and $\Inn_\R(W'') \cong G_\SM$.  We shall prove that all elements of $G' = \im \Phi_{e_1}$ preserve $e_1$ up to phase and have unit determinant on $W$.  Similarly, all elements $G'' = \im \Phi_{e_2}$ preserve $e_2$ up to phase and have unit determinant on $W'$.

First, we show that there is an injective homomorphism $\Phi_{e_1}$ that maps each element $g \in \Inn_\R(W')$ to an element $\Phi_{e_1}(g) \in \Inn_\R(W)$ that has unit determinant, with $\Phi_{e_1}(g)|_{W'}=g|_{W'}$ and $\Phi(g)(e_1)=\exp{(i\alpha)}e_1$ for some $\alpha \in \R$. We note that the inner automorphism $\exp(i \alpha D(e_1,e_1))$ acts as $e^{i\alpha}$ on the one-dimensional Peirce space $P_{1}(e_1)W$, as $e^{i \alpha/2}$ on the 10-dimensional space $P_{1/2}(e_1)W$, and as the identity on the 5-dimensional space $P_{0}(e_1)W$.  Thus 
\[  \textstyle{\det_W}(\exp(i \alpha D(e_1,e_1))=e^{6i\alpha} ,\]
so if for all $\beta \in \R$ we define
\begin{equation}\label{eqn:appendx:se1}
    s_{e_1}(\beta)=\exp(8i\beta D(e_1,e_1)-3i \beta)\in \Inn_\R(W)
\end{equation}
then we have
\[  s_{e_1}(\beta) w'=e^{i\beta}w' \]
for any $w'\in P_{1/2}(e_1)W$ and $s_{e_1}(\alpha+\beta)=s_{e_1}(\alpha)s_{e_1}(\beta)$. Thus, for all $g \in \Inn_\R W$ we have 
\[   \Phi_{e_1}(g) = \exp(-i\beta(g))s_{e_1}(\beta(g))g \]
where $\beta(g) \in \R$ is such that $\det_W(g) =e ^{i\beta(g)}$, where we are taking the determinant of the natural action of $g \in \Inn_\R(W')$ on $W$, which is based on the embedding of $W'$ in $W$ as a hermitian Jordan subtriple. It can be seen that this definition implies $\Phi_{e_1}(g)|_{W'}=g|_{W'}$, $\det_W \Phi_{e_1}(g)=1$ for any $g$, and that $\Phi_{e_1}$ is a smooth map, since $\beta(g)$ is a real valued smooth function and left multiplications by $\exp{iD(e_1,e_1)}$ and $\exp{(i\alpha)}$ are smooth as well. It is also a group homomorphism, since for any $g_1,g_2 \in \Inn_\R(W')$
\begin{equation*}
    \begin{split}
        \Phi_{e_1}(g_1)\Phi_{e_1}(g_2)
        &=\exp(-i\beta(g_1))s_{e_1}(\beta(g_1))g_1\exp(-i\beta(g_2))s_{e_1}(\beta(g_2))g_2
        \\&=\exp(-i\beta(g_1))\exp(-i\beta(g_2))s_{e_1}(\beta(g_1))s_{e_1}(\beta(g_2))g_1g_2
        \\&=\exp(-i(\beta(g_1)+\beta(g_2)))s_{e_1}(\beta(g_1)+\beta(g_2))g_1g_2
        \\&=\exp(-i(\beta(g_1g_2)))s_{e_1}(\beta(g_1g_2))g_1g_2=\Phi_{e_1}(g_1g_2),
    \end{split}
\end{equation*}
where we used the fact that for any $\beta \in \R$ $s_{e_1}(\beta)$ and $e^{i\beta}$ commute with all elements of $\in \Inn_\R(W')$.

By the Peirce multiplication rule \ref{eq:peirce_multiplication_rule}, any inner derivation of $W'$ applied to $e_1$ gives an element of \eqref{eq:peirce_multiplication_rule}
\begin{equation*}
    [P_{1/2}(e_1)W, \, P_{1/2}(e_1)W, \, P_{1}(e_1)W] \subset P_{1}(e_1)W.
\end{equation*}
Exponentiating such derivations, it follows that for all $g \in \Inn_\R(W')$ we have $g e_1 \in P_1(e_1)W$.  But the Peirce 1-space of $e_1$ is one-dimensional because $e_1$ is a minimal tripotent.  Thus we have 
\[  g(e_1)=e^{i\alpha} e_1\] 
for some $\alpha \in \R$, and it follows that
\begin{equation*}
\begin{split}
        \Phi_{e_1}(g)e_1&=\exp(-i\beta(g)) \, s_{e_1}(\beta(g)) \, g(e_1) 
        \\&=\exp(-i\beta(g)+i\alpha)\, s_{e_1}(\beta(g))e_1
        \\&=\exp(4i\beta(g)+i\alpha) e_1,
\end{split}
\end{equation*}
where we used the definition of the function $s_{e_1}$ in equation \eqref{eqn:appendx:se1} and the fact that $D(e_1,e_1)e_1=e_1$.
This map is injective, since $\Phi_{e_1}(g)|_{W'}=g|_{W'}\neq g'|_{W'}=\Phi_{e_1}(g')|_{W'}$ implies $\Phi_{e_1}(g)\neq \Phi_{e_1}(g')$. Thus, $\Phi_{e_1}$ is a group isomorphism between $\Inn_\R(W')\cong \U(5)$ and $G' = \im \Phi_{e_1}$.

Next, the map $\Phi_{e_2}$ is defined analogously to $\Phi_{e_1}$ by replacing in the above arguments $W$ by $W'$, $W'$ by $W''$ and using the corresponding inner automorphism groups $\Inn_\R (W')\cong \U(5)$ and $\Inn_\R (W'')\cong G_\SM$. The role of $e_1$ that induced the colinear subspace $W'$ of $W$ is also replaced by $e_2$ that induces the colinear subspace $P_{1/2}(e_2)W'=W''\subset W'$. Taking into account that $W'$ and its Peirce subspaces have different dimensions in this case, we replace $s_{e_1}$ by 
\begin{equation*}
    s_{e_2}(\beta)=\exp{(10i\beta D(e_2,e_2)-4i\beta)}.
\end{equation*}
Then, the definition of $\Phi_{e_2}$ follows that of $\Phi_{e_1}$:
\begin{equation*}
    \Phi_{e_2}(g)=s_{e_2}(\beta'(g))\exp(-i\beta'(g))g,
\end{equation*}
where $\beta'(g)$ is given by $\det_{W'}(g)=e^{i\beta'(g)}$. The same arguments as for $\Phi_{e_1}$ now also imply that $\Phi_{e_2}$ is smooth group homomorphism and bijective between $\Inn_\R(W'')\cong G_\SM$ and $G'' = \im\Phi_{e_2}$.

Thus, for any $g\in \Inn_\R(W'')$, the map $\Phi_{e_2}(g)$ preserves $e_2$ up to phase and has the same natural action on $W''$ as $g$. Consequently, for any $g\in\Inn_\R (W'')$ $\Phi(g)=\Phi_{e_1}\circ\Phi_{e_2}(g)\in \Inn_\R (W)$ has the same action on $W'$ as $\Phi_{e_2}(g)$ which has the same action on $W''\subset W'$ as $g$.  Thus $\det_W\Phi_{e_1}(\Phi_{e_2}(g))=1$ and $\det_{W'}\Phi_{e_1}(\Phi_{e_2}(g))=\det_{W'}\Phi_{e_2}(g)=1$. Additionally, the colinear tripotents $e_1\in W$ and $e_2\in W'\subset W$ are preserved up to phase by $\Phi_{e_1}\circ\Phi_{e_2}(g)$. The existence of such an isomorphism $\Phi$ means that there is a subgroup $G'' \cong G_\SM$ of $\Inn_\R(W)=(\Spin(10)\times \U(1))/\Z_{4}$ that satisfies the statements of the theorem.
\end{proof}

\bibliographystyle{unsrt} 
\bibliography{References}{}

\end{document}